\documentclass{article}
\usepackage[utf8]{inputenc}
\usepackage{amsthm}
\usepackage{amsmath}
\usepackage{amssymb}
\usepackage{mathtools}
\usepackage{geometry}
\usepackage{enumerate}
\usepackage{authblk}
\usepackage{color}
\usepackage{hyperref}

\newtheorem{definition}{Definition}
\newtheorem{lemma}{Lemma}
\newtheorem{theorem}{Theorem}
\newtheorem{proposition}{Proposition}
\theoremstyle{remark}
\newtheorem{remark}{Remark}

\newcommand\wt{\mathrm{wt}_2}
\newcommand\ftwo{\mathbb{F}_{2}}
\newcommand{\F}{\mathbb{F}}

\renewcommand\footnotemark{}
\title{Gold Functions and Switched Cube Functions Are Not 0-Extendable in Dimension $n > 5$\thanks{This version of the article has been accepted for publication, after peer review
but is not the Version of Record and does not reflect post-acceptance improvements, or any
corrections. The Version of Record is available online at: \url{https://doi.org/10.1007/s10623-022-01111-6}.}}
\author[1]{Christof Beierle}
\affil[1]{Faculty of Computer Science, Ruhr University Bochum, Universit\"atsstra\ss e 150, 44801 Bochum, Germany}
\author[2,3]{Claude Carlet}
\affil[2]{LAGA, University of Paris 8, Saint-Denis, France}
\affil[3]{Department of Informatics, University of Bergen, PB 7803, 5020 Bergen, Norway}
\date{}

\begin{document}
\maketitle
\begin{abstract}
    In  the independent works by Kalgin and Idrisova and by Beierle, Leander and Perrin, it was observed that the Gold APN functions  over $\F_{2^5}$ give rise to a quadratic APN function in dimension 6 having maximum possible linearity of $2^5$ (that is, minimum possible nonlinearity $2^4$). In this article, we show that the case of $n 
    \leq 5$ is quite special in the sense that Gold APN functions in dimension $n>5$ cannot be extended to quadratic APN functions in dimension $n+1$ having maximum possible linearity. In the second part of this work, we show that this is also the case for APN functions of the form $x \mapsto x^3 + \mu(x)$ with $\mu$ being a quadratic Boolean function.
\end{abstract}


\section{Introduction and Preliminaries}
Throughout this work, let $n \in \mathbb{N}^* \coloneqq \{1,2,\dots\}$ be a positive integer. We call \emph{$(n,n)$-functions}, or \emph{vectorial functions in $n$ variables},  the functions from $\ftwo ^n$ to itself. These functions can also be viewed as functions from $\F_{2^n}$ to itself through the choice of a basis of $\F_{2^n}$ as a vector space over $\ftwo$. We call \emph{component functions} of $F$ the nonzero linear combinations of its coordinate functions; in the case of a function from $\F_{2^n}$ to itself, the component functions equal $tr(vF(x))$, where $v\in \F_{2^n}$ is nonzero and $tr(x)=x+x^2+x^{2^2}+\dots +x^{2^{n-1}}$ is the absolute trace function over $\F_{2^n}$. The \emph{nonlinearity} $nl(F)$ of such a function $F$ equals the minimum Hamming distance between the component functions of $F$ and the affine Boolean functions (i.e., the linear combinations of the input coordinates and constant function 1  -- in the case of a function from $\F_{2^n}$ to itself, the affine functions have the form $tr(ux+a)$, $u,a\in \F_{2^n}$). We call the \emph{linearity} of $F$ the (non-negative) value $2^n-2 \cdot nl(F)$. It should be as low as possible for the function to allow the resistance to the linear attack of a block cipher using it as a substitution box, see e.g., \cite{carlet2021boolean}.

\emph{Almost perfect nonlinear (APN) functions} in $n$ variables are those $(n,n)$-functions such that every derivative $D_aF(x)=F(x)+F(x+a)$, $a\in \ftwo ^n$, $a\neq 0$, is 2-to-1 (i.e., is such that the preimage of any element of the codomain has size 0 or 2).
An important open problem mentioned in \cite{carlet2021boolean} on APN functions is to determine whether such functions can have low nonlinearity. 
All known infinite classes of APN functions have a rather high nonlinearity  (while there exist a few sporadic quadratic functions with lowest possible nonzero nonlinearity $2^{n-2}$, see \cite{beierle2021new}). It is not clear whether this property of the known infinite classes is because they are in fact peculiar (and would therefore not be representative of general APN functions) or whether there is a significant lower bound to be found on the nonlinearity of a large subclass of APN functions. There is a lower bound stronger than strict positivity, which is valid for a large subclass of APN functions that includes all known ones, see \cite{carlet2021properties}, but it is not significant. We consider that determining which is true between the two terms of the above alternative is currently the most important open problem on substitution boxes (vectorial functions) along with the ``big APN problem" (consisting in finding an APN permutation in an even number of variables at least 8). Indeed, if it happens that the known infinite classes are peculiar, it would mean that we know almost nothing on the whole class of APN functions. 

A way to tackle the question on the nonlinearity of APN functions is to first focus on quadratic functions in dimension $n \geq 3$, whose lowest nonzero nonlinearity is known. Recall that as proved in \cite{carlet2010vectorial}, all APN functions in dimension $n\geq 3$ have strictly positive nonlinearity. The lowest nonlinearity of quadratic APN functions then equals $2^{n-2}$ (and the largest linearity less than $2^n$ equals $2^{n-1}$).  A secondary construction of APN quadratic functions has been studied in \cite{DBLP:journals/corr/abs-2108-13280}, in which the functions are obtained as the extensions of APN quadratic functions in one less variable, and there is a way of imposing that the resulting APN function has nonlinearity $2^{n-2}$.
This results in the following special type of quadratic APN functions, called \emph{0-extendable functions}.

\begin{definition}[\cite{DBLP:journals/corr/abs-2108-13280}]
\label{def:0extendable}
Let $n \geq 3$. A quadratic APN function $F \colon \F_{2^n} \rightarrow \F_{2^n}$ is called \emph{0-extendable} if there exists a linearized polynomial $L \in \F_{2^n}[X]$ and a non-zero element $a \in \F_{2^n}$ such that the function
\begin{eqnarray*}
T:\F_{2^n} \times \ftwo &\to& \F_{2^n} \times \ftwo \\
\left(\begin{array}{c}
     x  \\
     y 
\end{array}\right) &\mapsto& 
\left(\begin{array}{c}
     F(x)  \\
     0 
\end{array}\right)+
\left(\begin{array}{c}
     L(x)  \\
     tr(a x) 
\end{array}\right)\cdot y
\end{eqnarray*}
is APN. We call such a function $T$ a
$(0,L,a)$-extension of $F$.
\end{definition}

Note that the property of being 0-extendable is invariant under EA-equivalence. We recall that two functions $F \colon V \rightarrow V$ and $G \colon W \rightarrow W$, with $V$ and $W$ being $n$-dimensional $\F_2$ vector spaces, are called \emph{EA-equivalent} if there exist affine bijections $A \colon W \rightarrow V$, $B \colon V \rightarrow W$ and an affine function $C \colon W \rightarrow W$ such that $G = B \circ F \circ A + C$. The following proposition outlines that a quadratic APN function with the theoretically highest possible value on its linearity is, as recalled above, necessarily obtained from a 0-extendable APN function in one dimension lower.

\begin{proposition}[\cite{DBLP:journals/corr/abs-2108-13280}]
\label{prop:extension}
Let $n \geq 3$. The linearity of a $(0,L,a)$-extension of a quadratic APN function $F \colon \F_{2^n} \rightarrow \F_{2^n}$ is equal to $2^{n}$. Moreover, any quadratic APN function $G \colon \F_{2^{n+1}} \rightarrow \F_{2^{n+1}}$ with linearity $2^n$ is EA-equivalent to some $(0,L,a)$-extension of a quadratic APN function in dimension $n$. 
\end{proposition}

An example of a 0-extendable function is the APN function $x \mapsto x^3$ over $\F_{2^5}$ with its $(0,X^{16}+X,1)$-extension belonging to the EA-equivalence class corresponding to the quadratic APN function in dimension $n=6$ with linearity 32, i.e., function no. 2.6 of~\cite{DBLP:journals/amco/EdelP09} (see~\cite[Section 5.3]{DBLP:journals/corr/abs-2108-13280}). A similar example can be given for the APN function $x \mapsto x^5$ over $\F_{2^5}$. Note that the fact that the quadratic APN monomial functions in dimension $n=5$ are 0-extendable was first observed in~\cite{DBLP:journals/iacr/KalginI20a} (although the authors did not name those functions 0-extendable). Besides that, up to EA-equivalence, four of the quadratic APN functions in dimension $n=7$ are 0-extendable~\cite{DBLP:journals/corr/abs-2108-13280}.

A natural question that arises from those examples is whether some other quadratic APN functions (possibly giving an infinite class) in larger dimension $n$ are 0-extendable, that is whether they give rise to quadratic APN functions with maximum possible linearity. The two examples in dimension $n=5$ above belong to the main known class of quadratic APN functions, that of \emph{Gold functions} $x \mapsto x^{2^t+1}$, with $\gcd(t,n)=1$, see~\cite{DBLP:journals/tit/Gold68,DBLP:conf/eurocrypt/Nyberg93}. The subclass corresponding to $t=1$ (i.e. the class of so-called {\em cube functions}) has been modified into a class of quadratic APN functions in~\cite{DBLP:journals/ffa/BudaghyanCL09} by the addition of a Boolean function, giving the infinite APN family $x \mapsto x^3 + tr(x^9)$. For $n\geq 7$, these {\em switched cube functions} are CCZ-inequivalent to any quadratic monomial function (and EA-equivalent to $x \mapsto x^3$ for $n=5$). 

In the present paper, we answer the question of the 0-extendability of these two main classes negatively:
\begin{theorem}
\label{thm:gold}
Let $n, t \in \mathbb{N}^*$ with $n >5$ and $\gcd(t,n)=1$. Let $F \colon \F_{2^n} \rightarrow \F_{2^n}, x \mapsto x^{2^t+1}$. Then, $F$ is not 0-extendable.
\end{theorem}
\begin{theorem}\label{th2}
\label{thm:trace}
Let $n \in \mathbb{N}^*$ with $n >5$ and let $\mu \colon \F_{2^n} \rightarrow \F_2$ be a quadratic Boolean function such that $F \colon \F_{2^n} \rightarrow \F_{2^n}, x \mapsto x^3 + \mu(x)$ is APN. Then, $F$ is not 0-extendable.
\end{theorem}

\begin{remark}
In~\cite{DBLP:journals/amco/EdelP09}, Edel and Pott studied a more general notion of switching, i.e., instead of adding a Boolean function $\mu \colon \F_{2^n} \rightarrow \F_2$ to an APN function over $\F_{2^n}$, they allowed to add $z \cdot \mu$ for an arbitrary non-zero constant $z \in \F_{2^n}$. We note that our Theorem~\ref{thm:trace} also covers this more general notion of switching. Indeed, for $n$ being odd, having $F \colon \F_{2^n} \rightarrow \F_{2^n}, x \mapsto x^3 + z \cdot \mu(x)$, multiplying with $z^{-1}$ yields $x \mapsto z^{-1}x^3 + \mu(x) =
(z^{-3^{-1}}x)^3 + \mu(x)$, where the last equality holds because $x \mapsto x^3$ is a permutation. By multiplying $x$ with $z^{3^{-1}}$, one observes that $F$ is EA-equivalent to $x^3 + \mu'(x)$ with $\mu'(x) \coloneqq
\mu(z^{3^{-1}}x)$. 
\end{remark}

To prove those results, we will use the characterization of 0-extendable functions given in Proposition~\ref{prop:ortho_characterization} below. For that, we need the notion of the ortho-derivative of a quadratic APN function.
\begin{definition}[Ortho-derivative~\cite{DBLP:journals/corr/abs-2103-00078}]
Let $F \colon \F_{2^n} \rightarrow \F_{2^n}$ be a quadratic APN function. The \emph{ortho-derivative} of $F$ is defined as the unique function $\pi_F \colon \F_{2^n} \rightarrow \F_{2^n}$ with $\pi_F(0) = 0$ such that, for all non-zero $a \in \F_{2^n}$, we have $\pi_F(a) \neq 0$ and
\begin{align*} \forall x \in \F_{2^n} \colon tr(\pi_F(a) B_{a}(x)) = 0\;,\end{align*}
where $B_{a} \colon \F_{2^n} \rightarrow \F_{2^n}$ is the function $x \mapsto F(x) + F(x+a) + F(a) + F(0)$.
\end{definition}

\begin{proposition}[\cite{DBLP:journals/corr/abs-2108-13280}]
\label{prop:ortho_characterization}
Let $n \geq 3$. A quadratic APN function $F \colon \F_{2^n} \rightarrow \F_{2^n}$ is 0-extendable if and only if there exist a linearized polynomial $L \in \F_{2^n}[X]$ and a non-zero element $a \in \F_{2^n}$ such that $tr(\pi_F(x)L(x)) = 1$ for all non-zero $x \in \F_{2^n}$ with $tr(ax) = 0$.
\end{proposition}

We will also need a series of lemmas. We shall only outline the proofs of Lemmas \ref{lem:gold_exponents2}, \ref{lem:gold_exponents3}, \ref{l3} and \ref{lem:trace2} in detail. The proofs of the other lemmas work similarly as the one of Lemma~\ref{lem:trace2}. All these lemmas and their consequences, i.e., the derivation of Relations (\ref{eq:second}), (\ref{eq:first}), (\ref{eq:fourth}), and (\ref{eq:fifth}), being rather technical, we have verified  with a computer the validity of each of those relations for some random choices of $L$ and of $a$ with $n\leq 13$.

\section{The Case of Gold Functions}
For a positive integer $j \in \mathbb{N}^*$, we denote by $\wt(j)$ the Hamming weight of the binary expansion of $j$, which is defined as the number of non-zero coefficients $a_i$ when $j$ is represented as $j = \sum_{i \in \mathbb{N}^*} a_i2^{i-1}, a_i \in \{0,1\}$, and we call it below indifferently the \emph{2-weight} or the \emph{weight of the binary expansion  of $j$}. We shall see in the proofs of Theorems \ref{thm:gold} and \ref{th2} that proving non-0-extendability by using Proposition \ref{prop:ortho_characterization} amounts to determining the cases in which some numbers can have 2-weight at most 2 (and addressing the resulting particular cases). For this, in each case, we will need preliminary lemmas.

The proof of Theorem~\ref{thm:gold} will be split into the three cases $t = 1$, $1 < t < \frac{n}{2}-1$, and $t = \frac{n-1}{2}$. To prove the second and third cases, we shall need the following lemma (as described above). 
\begin{lemma}
\label{lem:gold_exponents2}
Let $n,s,t \in \mathbb{N}^*, k \in \mathbb{N}$ with $k <n$, let $ 1< t < \frac{n}{2}$ and $s \in \{ 2,3,\dots,t-1,t+2,t+3,\dots,n-2\}$. Then, the binary expansion of $2^s-(2^t+1) + 2^k + 2^{k+t} \mod (2^n-1)$ has weight strictly greater than 2 if $k \notin \{0, t+1\}$.
\end{lemma}
\begin{proof}
Let us define $h \coloneqq 2^s - 2^t -1 + 2^k + 2^{k+t} \mod (2^n-1)$. We observe that, given any positive integer $j < n$, the number $h_j \coloneqq h \mod (2^j-1)$ satisfies $\wt(h_j) \leq \wt(h)$.  Indeed, the reduction of each power of 2 in the binary expansion of $h$ gives rise to at most one power of 2 in the binary expansion of $h_j$ (and groupings of powers of 2 being possible, this inequality may be strict). Therefore, to prove that the 2-weight of $h$ is strictly greater than 2, it is enough to find some $j < n$ such that $\wt(h_j) \geq 3$. We divide the proof into several cases. In our computations, we will excessively use the identity $2^j = 1 + \sum_{i=0}^{j-1} 2^i$ for $j \in \mathbb{N}^*$.

\paragraph{Case $s < t$.} 
For $k = 1$, we have $h = 2^s + 2^t + 1$. Since $s \notin \{0,t\}$, it follows that $\wt(h) \geq 3$. 

For $1 < k < t-1$, we first consider the case of $s \neq 2$. We consider
\[h_t = h \mod (2^t-1) = 2^s + 2^{k+1}  - 2 = 2^s + \sum_{i=1}^k 2^i,\]
and we can check in each of the two cases $s\leq k$ and $s>k$ that it has 2-weight (respectively equal to $s$ and $k+1$) greater than 2. 
For $s = 2$, we have \begin{align*}h &= 2^2 - 2^t + 2^k + \sum_{i=0}^{k+t-1}2^i = 2^2 + 2^k + \sum_{i=0}^{t-1} 2^i + \sum_{i=t+1}^{k+t-1}2^i \\
&= 2^2 + \sum_{i=0}^{k-1}2^i + \sum_{i=t}^{k+t-1}2^i = 1+2+2^k + \sum_{i=t}^{k+t-1}2^i ,\end{align*}
which has 2-weight at least 3.

For $k=t-1$, we have $h = 2^s-2^{t-1} + 2^{2t-1} - 1 = 2^s - 2^{t-1} + \sum_{i=0}^{2t-2} 2^i = 2^s + \sum_{i=0}^{t-2} 2^i + \sum_{i=t}^{2t-2}2^i$, which has 2-weight strictly greater than $2$ if $t>2$ (since $s < t$). In the case of $t=2$, we have $h = 2^s + 2^2 + 1$, which has 2-weight 3, since $s \notin \{0,t\}$, that is, $s \notin \{0,2\}$.

For $k = t$, we have $h = 2^s + 2^{2t} - 1$, whose binary expansion is $\sum_{i=0}^{s-1}2^i+2^{2t}$, since $s<2t<n$, and has then 2-weight at least $s+1$, and is greater than 2 since $s \geq 2$.

For $t+1 < k < n$, we have $k-1>t$. Then,
\[ h = \begin{cases} 2^s + \sum_{i=0}^{t-1}2^i + \sum_{i=t+1}^{k-1} 2^i + 2^{k+t} & \text{if } k+t < n \\
2^s + \sum_{i=0}^{t-1}2^i + \sum_{i=t+1}^{k-1} 2^i + 2^{k+t-n} & \text{if } k+t \geq n
\end{cases}.\]
In the first case ($k+t < n$), the binary expansion of $h$ is
\[\sum_{i=0}^{s-1}2^i + \sum_{i=t}^{k-1} 2^i + 2^{k+t},\]
which has weight $s+(k-t)+1\geq 3$. In the second case ($k+t  \geq n)$, we have $k+t-n <t$. If $k+t-n = 0$, then $h = 2^s + 2^t + \sum_{i=t+1}^{k-1} 2^i$. If $k+t-n = t-1$, then $h = 2^s + \sum_{i=0}^{t-2}2^i + \sum_{i=t}^{k-1} 2^i = \sum_{i=0}^{s-1}2^i + 2^{t-1} + \sum_{i=t}^{k-1}2^i$. If $k+t-n \notin \{0,t-1\}$, we have $h = 2^s + \sum_{i=0}^{k+t-n-1}2^i + \sum_{i=k+t-n+1}^{t-1}2^i$, which has binary representation $\sum_{i=0}^{s-1}2^i + \sum_{i=k+t-n}^{t-1}2^i$ if $s\leq k+t-n$, binary representation $\sum_{i=0}^{k+t-n-1}2^i + 2^t$ if $s = k+t-n+1$, and binary representation $\sum_{i=0}^{k+t-n-1}2^i + \sum_{i=k+t-n+1}^{s-1}2^i + 2^t$ if $s > k+t-n+1$. In all of those cases, the binary representation of $h$ has weight greater than 2.

\paragraph{Case $s > t+1$.}
In this case, we have
\[h = \sum_{i=0}^{t-1} 2^i + \sum_{i=t+1}^{s-1} 2^i + 2^k + 2^{k+t}.\]
Now, since $s \neq n-1$, we can have $\wt(h) \leq 2$ only if $k \in \{0,t+1\}$ or $k+t \in\{0,t+1\}$ (mod $n$). 

For $k+t = 0 \mod n$, we have $h = 2^s-2^t + 2^{n-t}$ and there are three possible cases:\begin{itemize}\item  $s<n-t$, then since $t<s$, $h$ has binary expansion $\sum_{i=t}^{s-1} 2^i + 2^{n-t}$, which has weight $s-t+1\geq 3$;\item $s=n-t$, then $h=2^{s+1}-2^t$ has binary expansion $\sum_{i=t}^s 2^i$, which has weight $s-t+1\geq 3$ as well;\item $s>n-t$, then since $t<n-t$, $h$ has binary expansion $\sum_{i=t}^{n-t-1} 2^i +2^s$, which has weight $n-2t+1$, which is greater than 3 unless $t=\frac{n-1}2$; but since we are in the case $k+t = 0 \mod n$, this latter case means $k = t+1$,
which is excluded by the statement of the lemma.
\end{itemize} 

For $k+t = t+1 \mod n$, i.e., $k=1 \mod n$, we have $h = 2^s + 2^t + 1$, which has 2-weight greater than two since $s \notin \{0,t\}$.
\end{proof}

For handling the case of $t = \frac{n-1}{2}$, we need the following additional lemma.
\begin{lemma}
\label{lem:gold_exponents3}
Let $n,s,t \in \mathbb{N}^*, k \in \mathbb{N}$ with $n \geq 5$ being odd, $k < n$, $ t = \frac{n-1}{2}$ and $s \in \{ 1,2,\dots,t, t+3,t+4,\dots,n-2\}$. Then, the binary expansion of $2^s-1 + 2^k + 2^{k+t} \mod (2^n-1)$ has weight at most 2 if and only if $k \in \{0, \frac{n+1}{2}\}$. 
\end{lemma}
\begin{proof}
Let 
\[ h \coloneqq 2^s -1 + 2^k + 2^{k+t} \mod (2^n-1) = \sum_{i=0}^{s-1} 2^i + 2^k + 2^{k+t}.\]
We first observe that, for $k=0$, we have $h = 2^s + 2^t$ and, for $k = \frac{n+1}{2} = t+1$, we have $h = 2^s + 2^{t+1}$. In both cases, $\wt(h) \leq 2$.

Since $s \notin \{0,n-1\}$, the case of $\wt(h) \leq 2$ can only happen if $k \in \{0,1\}$, or $k+t \in \{0,1\}$ (mod $n$). For $k = 1$, we have $h = 2^s + 2^{t+1} + 1$, which has 2-weight greater than 2 since $s \notin \{0,t+1\}$. The case $k + t = 0 \mod n$ corresponds to the case of $k = \frac{n+1}{2}$. For $k + t = 1 \mod n$, it is $k = t+2$ and thus $h = 2^s + 2^{t+2} +1$, which has 2-weight greater than 2 since $s \notin \{0,t+2\}$.
\end{proof}

We can now prove our first main result.
\begin{proof}[Proof of Theorem~\ref{thm:gold}]
It is well known that the ortho-derivative $\pi_F \colon \F_{2^n} \rightarrow \F_{2^n}$ of $F$ equals the function  $x \mapsto x^{2^n -2^t - 2}$. Indeed, given $a \in \F_{2^n} \setminus \{0\}$, we have $B_{a}(x)= F(x) + F(x+a) + F(a) + F(0)=a x^{2^t}+a^{2^t}x$ and therefore $tr(\beta B_{a}(x))=tr((\beta a)^{2^{-t}}+\beta a^{2^t})x)$. Since $\gcd(t,n)=1$, the expression $(\beta a)^{2^{-t}}+\beta a^{2^t}$ equals 0 if and only if $\beta \in \{0, a^{2^n -2^t - 2}\}$. According to Proposition \ref{prop:ortho_characterization} we need to show that there does not exist a linearized polynomial $L \in \F_{2^n}[X]$ such that $tr(x^{2^n -2^t - 2}L(x)) = 1$ for all non-zero $x \in \F_{2^n}$ with $tr(x) = 0$ (see Proposition 7 of~\cite{DBLP:journals/corr/abs-2108-13280}).

Let us assume that there exists a linearized polynomial $L = \sum_{j=0}^{n-1}\ell_jX^{2^j} \in \F_{2^n}[X]$, such that for every nonzero $x \in \F_{2^n}$, $tr(x)=0$ implies $tr(x^{2^n -2^t - 2}L(x))=1$, that is, the Boolean function \[f(x) \coloneqq  (tr(x)+1)(tr(x^{2^n-2^t-2}L(x))+1)\] equals the indicator of $\{0\}$, that is, equals $x^{2^n-1}+1$ in $\mathbb{F}_{2^n}[x]/(x^{2^n}+x)$.

For $r\in \{1,\dots ,2^n-2\}$, the coefficient of $x^r$ (supposed to be zero) in the univariate form of $f$ equals:\begin{enumerate}[(i)]
\item the sum $\sum_{(i,j,k) \in S} \ell_j^{2^k}$, where \[S = \{(i,j,k) \mid i,j,k\in \{0,\dots ,n-1\} \text{ and } 2^i+2^{j+k} \equiv r+2^k+2^{k+t} \mod (2^n-1) \},\]
\item plus $1$ if $r=2^i$ for some $i\in \{0,\dots ,n-1\}$,
\item plus the sum  $\sum_{(j,k) \in T}\ell_j^{2^k}$ where \[T=\{(j,k)\mid j,k\in \{0,\dots ,n-1\} \text{ and } 2^{j+k}\equiv r+ 2^k+2^{k+t} \mod (2^n-1)\}.\] 
\end{enumerate}

Note that $2^i+2^{j+k}$ has a binary expansion of weight 2 unless $j+k
\equiv i\mod n$, in which case it has a binary expansion of weight 1. 

Without loss of generality, we can assume that $t < \frac{n}{2}$. The reason is that the function $x \mapsto x^{2^{t'}+1}$ for $t' > \frac{n}{2}$ is  EA-equivalent to a function $x \mapsto x^{2^t+1}$ with $t < \frac{n}{2}$ and the property of being 0-extendable is invariant under EA-equivalence. Moreover, the case $t = \frac{n}{2}$ cannot occur since $n$ must necessarily be odd (if $n>2$ is even, any quadratic APN function $G \colon \F_{2^{n+1}} \rightarrow \F_{2^{n+1}}$ cannot have linearity $2^n$ since it must be almost bent, see~\cite{DBLP:journals/dcc/CarletCZ98}). The remainder of this proof is split into several cases.

\paragraph{Case $t=1$.}
Let $4\leq s\leq n-2$ (assuming $n\geq 6$) and $r=2^s-3=1+2^2+2^3+\dots +2^{s-1}$. Then $r+2^k+2^{k+t} \mod (2^n-1)$ equals:
\begin{align*}2^s &\text{ for } k=0,\\
1+2+2^s &\text{ for } k=1,\\ 1+2^3+2^s &\text{ for } k=2,\\ 1+2^2+\dots +2^{k-1}+2^{k+1}+2^s &\text{ for every } k\in \{3,\dots ,s-2\},\\ 1+2^2+\dots +2^{s-2}+2^{s+1} &\text{ for } k=s-1\\1+2^2+\dots +2^{s-1}+2^k+2^{k+1} &\text{ for every } k\in \{s,\dots ,n-2\}\\2+2^2+\dots +2^{s-1}+2^{n-1} &\text{ for } k=n-1.
\end{align*}
For $s=n-1$, the only case where the situation differs with respect to $s\leq n-2$ is for $k=s-1=n-2$ (for which $r+2^k+2^{k+t}\equiv 2+\dots +2^{n-3} \mod (2^n-1)$) and for $k=n-1$ (for which $r+2^k+2^{k+t}\equiv 2+\dots +2^{n-1} \mod (2^n-1)$).

Since the only case where we obtain an integer of 2-weight at most 2 is the first one, i.e., for $k=0$, we deduce that the coefficient of $x^r$ in the univariate form of $f$ equals $\ell_{s-1}+\ell_s$ (the second term in this sum coming from the case (iii) above). Hence, for the Boolean function $f$ to equal the indicator of $\{0\}$, we must have  $\ell_s=\ell_{s-1}$, for every $s=4,\dots ,n-1$, and therefore, $\ell_s=\ell_3$.

This yields
\begin{align*}L(x)=\ell_0x+\ell_1x^2+\ell_2x^4+\ell_3\sum_{i=3}^{n-1}x^{2^i}=(\ell_0+\ell_3)x+(\ell_1+\ell_3)x^2+(\ell_2+\ell_3)x^4+\ell_3 tr(x)\end{align*}
and  expressing that $tr(x)=0$ implies $tr(x^{2^n-4}L(x))=1$ for every $x\neq 0$ results now in the same property applied to the function $(\ell_0+\ell_3)x+(\ell_1+\ell_3)x^2+(\ell_2+\ell_3)x^4$. In other words, we can without loss of generality assume that $L$ has degree at most 4, that is, assume that $\ell_3=\ell_4=\dots =\ell_{n-1}=0$.

Let us now consider other values of $r$ in order to show the inexistence of such (necessarily nonzero) $L$ for $n\geq 6$: For $r=7$, we get \[\ell_1^2+\ell_0=0\] and for $r=1$, we obtain $\ell_1+\ell_0^{2^{n-1}}+\ell_2^{2^{n-1}}+1+\ell_2=0$, hence \[\ell_2+\ell_2^2=1.\] For $r=3$, we get $\ell_2+\ell_2^2=0$, a contradiction.

\paragraph{Case $1 < t < \frac{n}{2}-1$.}
From Lemma~\ref{lem:gold_exponents2}, it follows that for each $r = 2^s - (2^t+1) \mod (2^n-1)$ with $s \in \{2,3,\dots,t-1,t+2,t+3,\dots,n-2\}$, the binary expansion of $r + 2^k + 2^{k+t} \mod (2^n-1)$ can have weight at most 2 only if $k \in \{0,t+1\}$. For $k = t+1$, we have 
\[r + 2^k + 2^{k+t} \mod (2^n-1) = 2^s + \sum_{i=0}^{t-1}2^i + 2^{2t+1}, \]
which has 2-weight at most 2 only if $2t+1 = n$, i.e., only if $t = \frac{n-1}{2}$. 
Therefore, we only need to consider $k=0$ and we have $r + 2^k + 2^{k+t} \equiv 2^s \mod (2^n-1)$. Hence, the coefficient of $x^r$ in the univariate form of $f$ equals $\ell_{s-1} + \ell_s$. We deduce that, for the Boolean function $f$ to equal the indicator of $\{0\}$, we must have  $\ell_s=\ell_{s-1}$, for every $s \in \{2,3,\dots,t-1,t+2,t+3,\dots,n-2\}$ and therefore, 
\begin{align*}
    L(x) &= \ell_0 x + \ell_1 \sum_{i=1}^{t-1}x^{2^i} + \ell_t x^{2^t} + \ell_{t+1}\sum_{i=t+1}^{n-2} x^{2^i} + \ell_{n-1}x^{2^{n-1}} \\
    &= (\ell_0 + \ell_{t+1})x + (\ell_1+\ell_{t+1})\sum_{i=1}^{t-1}x^{2^i} + (\ell_t + \ell_{t+1}) x^{2^t} + (\ell_{t+1} + \ell_{n-1})x^{2^{n-1}} + \ell_{t+1} tr(x).
\end{align*}
Expressing that $tr(x)=0$ implies $tr(x^{2^n-2^t-2}L(x))=1$ for every $x\neq 0$ results now in the same property applied to the function $(\ell_0 + \ell_{t+1})x + (\ell_1+\ell_{t+1})\sum_{i=1}^{t-1}x^{2^i} + (\ell_t + \ell_{t+1}) x^{2^t} + (\ell_{t+1} + \ell_{n-1})x^{2^{n-1}}$. In other words, we can without loss of generality assume that $L(x)$ is of the form
\[ \ell_0x + \ell_1\sum_{i=1}^{t-1}x^{2^i} + \ell_t x^{2^t} +  \ell_{n-1}x^{2^{n-1}}.\]

Again, let us consider other values of $r$ to show the inexistence of such (necessarily nonzero) $L$: For $r = 7$, we get $\ell_t^{2^t} + \ell_0 = 0$ and for $r=1$, we obtain $\ell_0 + \ell_1^{2^t} + \ell_t^{2^t} = 1$, hence $\ell_1 = 1$. For $r=2^t+1$, we obtain $\ell_1 = 0$, which is a contradiction.

\paragraph{Case $t = \frac{n-1}{2}$.} From Lemma~\ref{lem:gold_exponents2}, it follows that for each $r = 2^s - (2^t+1) \mod (2^n-1)$ with $s \in \{2,3,\dots,t-1,t+2,t+3,\dots,n-2\}$, the binary expansion of $r + 2^k + 2^{k+t} \mod (2^n-1)$ has weight at most 2 only if $k \in \{0,t+1\}$. In that case, we have $r + 2^k + 2^{k+t} \equiv 2^s\mod (2^n-1)$ for $k=0$, and  $r + 2^k + 2^{k+t} \equiv 2^s + 2^{t}\mod (2^n-1)$ for $k=t+1$. Therefore, the coefficient of $x^r$ in the univariate form of $f$ equals $\ell_{s-1} + \ell_{n-1}^{2^{t+1}} + \ell_{s-t-1 \mod n}^{2^{t+1}} + \ell_s$, which yields
\begin{equation}\label{eq:s}\forall s \in \{2,3,\dots,t-1,t+2,t+3,\dots,n-2\} \colon \quad \ell_{s-1} + \ell_{n-1}^{2^{t+1}} + \ell_{t+s}^{2^{t+1}} + \ell_s = 0.\end{equation}

From Lemma~\ref{lem:gold_exponents3}, it follows that for each $r = 2^s - 1$ with $s \in \{1,2,\dots,t,t+3,t+4,\dots,n-2\}$, the binary expansion of $r + 2^k + 2^{k+t} \mod (2^n-1)$ has weight at most 2 if and only if $k \in \{0,t+1\}$. In that case, we have $r + 2^k + 2^{k+t} \equiv 2^s + 2^t\mod (2^n-1)$ for $k=0$ and  $r + 2^k + 2^{k+t} \equiv 2^s + 2^{t+1}\mod (2^n-1)$ for $k=t+1$. Therefore, the coefficient of $x^r$ in the univariate form of $f$ equals $\ell_t + \ell_{s} + \ell_{0}^{2^{t+1}} + \ell_{t+s}^{2^{t+1}}$ for $s\notin \{1, t\}$ and $\ell_t + \ell_{t+1} +  \ell_{0}^{2^{t+1}} + \ell_{n-1}^{2^{t+1}}$  for $s=t$ and $\ell_t + \ell_{1} + \ell_{0}^{2^{t+1}} + \ell_{t+1}^{2^{t+1}} + 1$ for $s=1$, which yields
\begin{equation}\label{eq:s2}\forall s \in \{2,\dots,t-1,t+3,t+4,\dots,n-2\} \colon \quad \ell_t + \ell_{s} + \ell_{0}^{2^{t+1}} + \ell_{t+s}^{2^{t+1}} = 0\end{equation}
and
\begin{equation}\label{eq:st}\ell_t  + \ell_{t+1} + \ell_{0}^{2^{t+1}} + \ell_{n-1}^{2^{t+1}} = 0\end{equation}
and 
\begin{equation}\label{eq:st1}\ell_t + \ell_{1} + \ell_{0}^{2^{t+1}} + \ell_{t+1}^{2^{t+1}} = 1.\end{equation}

Combining Equation~(\ref{eq:s}) and Equation~(\ref{eq:s2}), we obtain
\begin{equation*}
    \forall s \in \{2,\dots,t-1,t+3,t+4,\dots,n-2\} \colon \quad \ell_t + \ell_0^{2^{t+1}} + \ell_{n-1}^{2^{t+1}} + \ell_{s-1} = 0,
\end{equation*}
hence $\ell_{s} = \ell_{t+1}$ for all $s \in \{1,2,\dots,t-2,t+2,t+3,\dots,n-3\}$ (from Equation~(\ref{eq:st})). Further, from Equation~(\ref{eq:s}) and $s=t+2$, we obtain $\ell_{t+1} + \ell_{n-1}^{2^{t+1}} + \ell_{1}^{2^{t+1}} + \ell_{t+2} = 0$, hence $ \ell_{n-1} = \ell_{t+1}$. We observe that Equations~(\ref{eq:st}) and~(\ref{eq:st1}) are in contradiction.
\end{proof}

\section{The Case of Switched Cube Functions}
To show Theorem~\ref{thm:trace}, we also use the characterization of 0-extendable functions by their ortho-derivative. The ortho-derivative $\pi_F$ of the APN function $F$ we are considering is given in the following lemma. 

\begin{lemma}\label{l3}
Let $\mu \colon \F_{2^n} \rightarrow \F_2$ be a quadratic Boolean function such that $F \colon \F_{2^n} \rightarrow \F_{2^n}, x \mapsto x^3 + \mu(x)$ is APN. For the ortho-derivative $\pi_F$ of $F$, we have
\[ \pi_F(x) = (tr(x^{2^n-4})+1)x^{2^n-4} + tr(x^{2^n-4})\zeta(x)\]
for some function $\zeta \colon \F_{2^n} \rightarrow \F_{2^n}$.
\end{lemma}
\begin{proof}
Since $\zeta$ can be an arbitrary function over $\F_{2^n}$, we only need to show that $\pi_F(x) = x^{2^n-4}$ if  $x^{2^n-4}$ has trace zero.
We have $B_{a}(x) = ax^2 + a^2x + \mu_{a}(x)$, where $\mu_a(x) = \mu(x) + \mu(x+a) + \mu(a) + \mu(0)$. Let $a \in \F_{2^n} \setminus \{0\}$ with $tr(a^{2^{n}-4})=0$. Then,  \begin{align*}tr(a^{-3}B_a(x)) = tr(a^{-2}x^2 + a^{-1}x) + tr(a^{-3}\mu_a(x)) = tr(a^{-3}\mu_a(x)) = tr(a^{2^n-4}\mu_a(x)) = 0,\end{align*} where the last equality holds because the image of $\mu_a$ is contained in $\F_2$ and thus $tr(a^{2^n-4}\mu_a(x)) =  tr(a^{2^n-4})tr(\mu_a(x))$. Therefore, $\pi_F(a) = a^{-3} = a^{2^{n}-4}$. For $a = 0$, we trivially have $\pi_F(0) = 0 = 0^{2^n-4}$.
\end{proof}

\begin{remark}
It is shown in~\cite{DBLP:conf/itw/BudaghyanCH11} that for the APN family $F \colon x \mapsto x^3 + tr(x^9)$, we have $\zeta(x) = x^6 + x^{2^{n-1}+1} + x^{2^n - 3\cdot 2^{n-2} -1}$. More precisely, the authors showed that the Boolean function whose support equals the set of pairs $(a,b)$ such that $a\neq 0$ and the equation $F(x)+F(x+a)=b$ has solutions equals $1+tr\big(a^{-3}b+1\big)$  when $tr(a^{-3})=0$ and equals $tr\big(\zeta(a)b\big)+1$ otherwise.
\end{remark}

In~\cite{DBLP:journals/amco/EdelP09}, the authors classified all switched cube APN functions up to dimension $n=9$ (where they considered the more general notion of switching, i.e., functions of the form $F(x)+z \cdot \mu(x)$ for an APN function $F$, an element $z \in \F_{2^n} \setminus \{0\}$, and a Boolean function $\mu$). 
In~\cite{DBLP:journals/corr/abs-2108-13280}, the authors classified all quadratic 0-extendable APN functions in dimension $n=7$ with a computational approach and they also verified that none of the known quadratic APN functions in dimension $n=9$ is 0-extendable. In particular, they verified that APN functions of the form $x \mapsto x^3 + \mu(x)$, for a quadratic Boolean function $\mu$, are not 0-extendable for $n\in \{7,9\}$. In our proof, we can therefore assume that $n$ is an odd integer with $n \geq 11$ (we recall that a 0-extendable function can only exist in odd dimension).

To show that a switched cube function $F$ is not 0-extendable, we need to show that for any linearized polynomial $L  \in \F_{2^n}[X]$ and any non-zero element $a \in \F_{2^n}$, the Boolean function \[x \mapsto (tr(x^{2^n-4})+1)(tr(ax)+1)(tr(x^{2^n-4}L(x))+1) + tr(x^{2^n-4})(tr(ax)+1)(tr(\zeta(x) L(x))+1),\]
where $\zeta$ is the function as given in Lemma~\ref{l3}, does not equal the indicator of $\{0\}$. For that, it suffices to show that, for any linearized polynomial $L  \in \F_{2^n}[X]$ and any non-zero element $a \in \F_{2^n}$, the Boolean function
\begin{equation}\label{eq:f}f(x) \coloneqq (tr(x^{2^n-4})+1)(tr(ax)+1)(tr(x^{2^n-4}L(x))+1)\end{equation}
does not equal the indicator of $\{0\}$ (if this happens for all such $L$ and $a$,  then  $F$ is clearly not 0-extendable).

{\bf Observation 1}: Let $a \in \F_{2^n} \setminus \{0\}$ and $L(x) = \sum_{j=0}^{n-1}\ell_jX^{2^j} \in \F_{2^n}[X]$. Let $\sum_{r=0}^{2^n-1} u_r x^r$ be the  univariate representation of the function $$g \colon x \mapsto (tr(ax)+1)(tr(x^{2^n-4}L(x))+1).$$ 
Similarly to as we saw already for the case of $a=1$, for $r \in \{1,\dots,2^n-2\}$, the coefficient $u_r$ of $x^r$ equals:
\begin{enumerate}[(i)]
\item the sum $\sum_{(i,j,k) \in S} a^{2^i}\ell_j^{2^k}$, where \[S = \{(i,j,k) \mid i,j,k\in \{0,\dots ,n-1\} \text{ and } 2^i+2^{j+k} \equiv r+2^k+2^{k+1} \mod (2^n-1) \},\]
\item to which we add $a^{2^i}$ if $r=2^i$ for some $i\in \{0,\dots ,n-1\}$,
\item to which we add $\sum_{(j,k) \in T}\ell_j^{2^k}$, where \[T = \{(j,k) \mid j,k\in \{0,\dots ,n-1\} \text{ and }2^{j+k}\equiv r+2^k+2^{k+1}\mod (2^n-1)\}.\]
\end{enumerate}

We then have 
\begin{align*}
f(x) &= tr(x^{2^n-4})g(x) + g(x) = {\Big (}\sum_{s=0}^{n-1}x^{2^{n+s}-2^{s+2}}{\Big )}g(x) + g(x) \\
&= \sum_{r=0}^{2^n-1} \left( \sum_{s=0}^{n-1} u_r x^{r+2^{n+s} - 2^{s+2}}\right) + \sum_{r=0}^{2^n-1}u_rx^r = \sum_{r=0}^{2^n-1} \left( \sum_{s=0}^{n-1} u_{r-2^{s}+2^{s+2}} x^{r}\right) + \sum_{r=0}^{2^n-1}u_rx^r \\
&= \sum_{r=0}^{2^n-1}\left(u_r + \sum_{s=0}^{n-1} u_{r+2^{s}+2^{s+1}}\right) x^r,
\end{align*}where the exponents and indices are taken modulo $2^n-1$, 
so the coefficient $v_r$ of $x^r$ in the univariate representation of $f$ equals
\begin{equation}\label{vr} v_r = u_r + \sum_{s=0}^{n-1} u_{r+2^{s}+2^{s+1}}.\end{equation}

Similarly as we did in the case of the Gold functions, we will derive equations in the coefficients $\ell_i$ of $L$ by determining the coefficients $v_r$ for some well-chosen values of $r \in \{1,\dots,2^n-2\}$. Let us start with studying the coefficients $v_r$ with $r = 2^m-9$ for $5 \leq m \leq n-3$. For this, we need the following lemma.

\begin{lemma}\label{lem:trace2}
Let $n,m \in \mathbb{N}^*, s, k \in \mathbb{N}, s,k <n$ with $n > 9$, and $5 \leq m < n-2$, and $k \notin \{0,1\}$. Then, the binary expansion of $h \coloneqq 2^m-9 + 2^s + 2^{s+1} + 2^k + 2^{k+1} \mod (2^n-1)$ has weight strictly greater than 2. Moreover,
\begin{enumerate}
\item if $k=0$, the weight of $h$ is at most 2 only if $s=1$. In that case, we have $h = 2^{m-1} + 2^{m-1}$.
\item if $k=1$, the weight of $h$ is at most 2 only if $s=0$. In that case, we have $h = 2^{m-1} + 2^{m-1}$.
\end{enumerate}

Moreover, for each $0 \leq k <n$, the binary expansion of $2^m-9 + 2^k + 2^{k+1} \mod (2^n-1)$ has weight strictly greater than 2.
\end{lemma}
\begin{proof}
Let us first prove the statement on the weight of the binary expansion of $h$. We have:
\[ h = 1 + 2 + 2^2 + \sum_{i=4}^{m-1}2^i + 2^k + 2^{k+1} + 2^s + 2^{s+1} \mod (2^n-1).\]
In the following, we assume that $k \leq s$. The case of $k>s$ can be proven similarly by exchanging the roles of $s$ and $k$. Let us first consider the case of $k \in \{0,1\}$.

If $k = 0$, we have $h = 2 + \sum_{i=3}^{m-1}2^i + 2^s + 2^{s+1} \mod (2^n-1)$, which has 2-weight at least 3 if $s \neq 1$. Indeed, $s = 0$ implies that $h = 1 + 2^2 + \sum_{i=3}^{m-1}2^i $, $s= 2$ implies that $h=2+4+2^m$, $s=3$ implies $h=2+2^4+2^m$, and $s > 3$ implies that the binary expansion of $h$ contains $2+8$ and at least one additional power of 2 (since $m < n-2$). For $s=1$, we have $h = 2^m$.

If $k = 1$, we have $h = 1 + \sum_{i=2}^{m-1}2^i + 2^s + 2^{s+1} \mod (2^n-1)$, which has 2-weight at least 3 since $s \neq 0$ (note that we assumed $k \leq s$). In the remainder of this proof, let us assume that $k \notin \{0,1\}$.

\paragraph{Case $s < n-2$.} In that case, $2^k + 2^{k+1} + 2^{s} + 2^{s+1} \leq 2^{n-2} + 2^{n-1}$. Hence, since $s\geq 2$ and $k \geq 2$ (this is the case as $k \notin\{0,1\}$ and $k \leq s$ by assumption), $2^k + 2^{k+1} + 2^{s} + 2^{s+1}$ can combine only with $2^2 +\sum_{i=4}^{m-1}2^i$ while $1 + 2$ is kept the same, the $2$-weight of $h$ is then at least $3$.

\paragraph{Case $s = n-2$.} If $k \leq n-5$, we obtain $2^k + 2^{k+1} + 2^{s} + 2^{s+1} \leq 2^{n-3} + 2^{n-2} + 2^{n-1}$ and we can use the same argument as above. For $k = n-4$, we have $h = 1 + 2 + 2^2 + \sum_{i=4}^{m-1}2^i + \sum_{i=n-4}^{n-1} 2^i \mod (2^n-1)$, which can have 2-weight at most 2 only if $m=n-3$. Then, $h = \sum_{i=3}^{n-5}2^i$, which has 2-weight at least 3 since $n > 9$. For $k \in \{n-3,n-2\}$, it is straightforward to deduce that the 2-weight of $h$ is greater than 2.

\paragraph{Case $s = n-1$.} In this case, we have $h = \sum_{i=3}^{m-1}2^i + 2^k + 2^{k+1} + 2^{n-1}$, which has 2-weight at least 3.

Finally, the last statement of the lemma is clear, since the weight of the binary expansion of
\[ h' \coloneqq 2^m - 9 + 2^{k} + 2^{k+1} =  1 + 2 + 2^2 + \sum_{i=4}^{m-1}2^i + 2^k + 2^{k+1} \mod (2^n-1)\]
is strictly greater than 2 for all values of $k$.
\end{proof}

Case 1 in Lemma~\ref{lem:trace2} corresponds to Case (i) in Observation 1, with $i=j+k=m-1$ and to Case (iii) with $j+k=m$. Also Case 2 in Lemma~\ref{lem:trace2} corresponds to (i) with $i=j+k=m-1$ and to (iii) with $j+k=m$. Let now $5 \leq m \leq n-3$ and $r = 2^m-9$. From Lemma~\ref{lem:trace2}, we obtain that the coefficient $v_r$ of $x^r$ in the univariate representation of $f$ equals Expression (\ref{vr}) in which the sum in $s$ is in fact reduced to $s\in \{0,1\}$, i.e.,
\[v_r = u_{r+1+2} + u_{r+2+4}.\]

Further, from Lemma~\ref{lem:trace2}, we deduce that $u_{r+2+4}$, which corresponds to Case 1, equals $a^{2^{m-1}}\ell_{m-1} + \ell_m$ (where the last term comes from (iii)) and that $u_{r+1+2}$, which corresponds to Case 2, equals $a^{2^{m-1}} \ell_{m-2}^2 + \ell_{m-1}^2$. Hence, we obtain the relation
\begin{equation}\label{eq:second} v_{2^m-9} =\ell_m + \ell_{m-1}^2 + a^{2^{m-1}}\ell_{m-1} + a^{2^{m-1}}\ell_{m-2}^2.\end{equation}

Let us now study the coefficients $v_r$ with $r = 2^m-3$ for $5 < m \leq n-3$.
\begin{lemma}\label{lem:trace}
Let $n,m \in \mathbb{N}^*, s, k \in \mathbb{N}, s,k <n$ with $n \geq 9$, and $5 < m \leq n-3$, and $k \notin \{0,m-1,m,n-1\}$. Then, the binary expansion of $h \coloneqq 2^m-3 + 2^s + 2^{s+1} + 2^k + 2^{k+1} \mod (2^n-1)$ has weight strictly greater than 2. Moreover,
\begin{enumerate}
\item if $k=0$, the weight of $h$ is at most 2 only if $s=m$ or $s=m-1$. In that case, we have $h = 2^{m+1} + 2^{m+1}$ and $h = 2^{m-1} + 2^{m+1}$, respectively.
\item if $k=m-1$, the weight of $h$ is at most 2 only if $s=0$. In that case, we have $h = 2^{m-1} + 2^{m+1}$.
\item if $k=m$, the weight of $h$ is at most 2 only if $s=0$. In that case, we have $h = 2^{m+1} + 2^{m+1}$.
\item if $k=n-1$, the weight of $h$ is at most 2 only if $s=k=n-1$. In that case, we have $h = 2^{m-1} + 2^{m-1}$.
\end{enumerate}

Moreover, for $0 \leq k <n$, the binary expansion of $h' \coloneqq 2^m-3 + 2^k + 2^{k+1} \mod (2^n-1)$ has weight at most 2 if and only if $k=0$. In that case, $h' = 2^{m-1} + 2^{m-1}$.
\end{lemma}

Case 1 in Lemma~\ref{lem:trace} corresponds to Case (i) in Observation 1, which contributes in the coefficient $u_r$ of $x^r$ for the sum $\sum_{(i,j,k) \in S} a^{2^i}\ell_j^{2^k}$, where $$S = \{(i,j,k) \mid i,j,k\in \{0,\dots ,n-1\} \text{ and } 2^i+2^{j+k} \equiv r+2^k+2^{k+1} \mod (2^n-1) \},$$ with $i=j+k=m+1$ and with $i=m-1$, $j+k=m+1$ and with $i=m+1$, $j+k=m-1$; and to Case (iii), which contributes for the sum $\sum_{(j,k) \in T}\ell_j^{2^k}$, where $$T = \{(j,k) \mid j,k\in \{0,\dots ,n-1\} \text{ and }2^{j+k}\equiv r+2^k+2^{k+1}\mod (2^n-1)\},$$ with $j+k=m+2$. 

Case 2 corresponds to (i) with $i=m-1$, $j+k=m+1$ and with $i=m+1$, $j+k=m-1$, and to (ii), which adds $a^{2^i}$ if $r=2^i$.  

Case 3 corresponds to (i) with $i=j+k=m+1$, to (ii) and to (iii) with $j+k=m+2$. Case 4  corresponds to (i) with $i=j+k=m-1$ and to (iii) with $j+k=m$.

Let now $6 \leq m \leq n-3$ and $r = 2^m-3$. From Lemma~\ref{lem:trace}, we obtain that the coefficient $v_r$ of $x^r$ in the univariate representation of $f$ equals Expression (\ref{vr}) in which the sum in $s$ is in fact reduced to $s\in \{0,m-1,m,n-1\}$, i.e.,
\[v_r = u_r + u_{r+1+2} + u_{r+2^{m-1}+2^m} + u_{r+2^m+2^{m+1}} + u_{r+2^{n-1}+1}.\]
Further, from Lemma~\ref{lem:trace}, we deduce that $u_{r+1+2}$, which corresponds to Cases 2 and 3, equals $a^{2^{m-1}}\ell_2^{2^{m-1}} + a^{2^{m+1}}\ell_0^{2^{m-1}} + a^{2^{m+1}}\ell_1^{2^{m}} + a^{2^m} + \ell_2^{2^m}$ (where the last term comes from (iii) and the second last term from (ii)), $u_{r+2^{m-1} + 2^m}$, which corresponds to Case 1, equals $ a^{2^{m-1}}\ell_{m+1} + a^{2^{m+1}}\ell_{m-1}$, $u_{r+2^m+2^{m+1}} $, which corresponds to Case 1, equals $ a^{2^{m+1}}\ell_{m+1} + \ell_{m+2}$ (where the last term comes from (iii)), and $u_{r+2^{n-1}+1}$, which corresponds to Case 4, equals $ a^{2^{m-1}}\ell_m^{2^{n-1}} + \ell_{m+1}^{2^{n-1}}$. Further, $u_r = a^{2^{m-1}}\ell_{m-1} + \ell_m$ (where the last term comes from (iii)).
Hence, we obtain the relation
\begin{equation}
\label{eq:first}
\begin{split} 
v_{2^m-3} &= \ell_{m+2} + \ell_{m+1}^{2^{n-1}} + (a^{2^{m-1}}+a^{2^{m+1}})\ell_{m+1} + a^{2^{m-1}}\ell_m^{2^{n-1}} + \ell_m + (a^{2^{m-1}}+a^{2^{m+1}})\ell_{m-1}  \\
&+ a^{2^{m-1}}\ell_2^{2^{m-1}} + \ell_2^{2^m} + a^{2^{m+1}}\ell_1^{2^{m}} + a^{2^{m+1}}\ell_0^{2^{m-1}} + a^{2^m}.
\end{split}
\end{equation}

Let us continue by studying the coefficients $v_r$ with $r = 2^m-5$ for $3 < m \leq n-3$.
\begin{lemma}\label{lem:trace4}
Let $n,m \in \mathbb{N}^*, s, k \in \mathbb{N}, s,k <n$ with $n \geq  9$, and $3 < m \leq n-3$, and $k \notin \{0,1,m\}$. Then, the binary expansion of $h \coloneqq 2^m-5 + 2^s + 2^{s+1} + 2^k + 2^{k+1} \mod (2^n-1)$ has weight strictly greater than 2. Moreover, 
\begin{enumerate}
\item if $k=0$, the weight of $h$ is at most 2 only if $s=0$ or $s=1$. In that case, we have $h=2^m + 2^0$ and $h=2^m + 2^2$, respectively.
\item if $k=1$, the weight of $h$ is at most 2 only if $s=0$ or $s=m$. In that case, we have $h=2^m + 2^2$ and $h=2^{m+2} + 2^0$, respectively.
\item if $k=m$, the weight of $h$ is at most 2 only if $s=1$. In that case, we have $h = 2^{m+2} + 2^{0}$.
\end{enumerate}

Moreover, for $0 \leq k <n$, the binary expansion of $h' \coloneqq 2^m-5 + 2^k + 2^{k+1} \mod (2^n-1)$ has weight at most 2 if and only if $k=1$. In the case of $k=1$, we have $h' = 2^m + 2^0$.
\end{lemma}

Let $4 \leq m \leq n-3$ and $r = 2^m-5$. From Lemma~\ref{lem:trace4}, we obtain that the coefficient $v_r$ of $x^r$ in the univariate representation of $f$ equals Expression (\ref{vr}) in which the sum in $s$ is in fact reduced to $s\in \{0,1,m\}$, i.e.,
\[v_r = u_r + u_{r+1+2} + u_{r+2+4} + u_{r+2^m+2^{m+1}}.\]

Further, from Lemma~\ref{lem:trace4}, we deduce that $u_r = a^{2^m}\ell_{n-1}^2 + a\ell_{m-1}^2$, $u_{r+1+2} = a^{2^m}\ell_0 + a \ell_m + a^{2^m}\ell_1^2 + a^{2^2}\ell_{m-1}^2$, $u_{r+2+4} = a^{2^m}\ell_2 + a^{2^2}\ell_m + a^{2^{m+2}} \ell_{n-m}^{2^m} + a\ell_2^{2^m}$, and $u_{r+2^m + 2^{m+1}} = a^{2^{m+2}}\ell_{n-1}^2 + a\ell_{m+1}^2$. Hence, we obtain the relation
\begin{equation}
\label{eq:fourth}
\begin{split}
v_{2^m-5} &= (a^{2^{m+2}}+a^{2^m})\ell_{n-1}^2 + a^{2^{m+2}} \ell_{n-m}^{2^m} + a\ell_{m+1}^2 + (a^{4}+a)\ell_m + (a^{4}+a)\ell_{m-1}^2 \\ &+  a\ell_2^{2^m} + a^{2^m}\ell_2  + a^{2^m}\ell_1^2 + a^{2^m}\ell_0.
\end{split}
\end{equation}

Finally, we consider the coefficient $v_{19}$.
\begin{lemma}\label{lem:trace5}
Let $n \in \mathbb{N}^*, s, k \in \mathbb{N}, s,k <n$ with $n \geq  9$ and $k \notin \{0,2,n-1\}$. Then, the binary expansion of $h \coloneqq 19 + 2^s + 2^{s+1} + 2^k + 2^{k+1} \mod (2^n-1)$ has weight strictly greater than 2. Moreover,
\begin{enumerate}
\item if $k=0$, the weight of $h$ is at most 2 only if $s=2$. In that case, we have $h=2^5 + 2^1$.
\item if $k=2$, the weight of $h$ is at most 2 only if $s=0$ or $s=n-1$. In that case, we have $h=2^5 + 2^1$ and $h=2^{5} + 2^{n-1}$, respectively.
\item if $k=n-1$, the weight of $h$ is at most 2 only if $s=2$. In that case, we have $h = 2^{5} + 2^{n-1}$.
\end{enumerate}

Moreover, for each $0 \leq k <n$, the binary expansion of $19 + 2^k + 2^{k+1} \mod (2^n-1)$ has weight strictly greater than 2.
\end{lemma}

From Lemma~\ref{lem:trace5}, we obtain that the coefficient $v_{19}$ of $x^r$ in the univariate representation of $f$ equals Expression (\ref{vr}) in which the sum in $s$ is in fact reduced to $s\in \{0,2,n-1\}$, i.e.,
\[v_r = u_{r+1+2} + u_{r+4+8} + u_{r+2^{n-1}+1}.\]

Further, from Lemma~\ref{lem:trace5}, we deduce that $u_{r+1+2} = a^{2^5}\ell_{n-1}^{2^2} + a^2\ell_3^{2^2}$, $u_{r+4+8} = a^{2^5} \ell_1 + a^2 \ell_5 + a^{2^5}\ell_0^{2^{n-1}} + a^{2^{n-1}}\ell_6^{2^{n-1}}$, and $u_{r+2^{n-1} + 1} = a^{2^5}\ell_{n-3}^{2^2} + a^{2^{n-1}}\ell_3^{2^2}$. Hence, we obtain the relation
\begin{equation}
    \label{eq:fifth}
    v_{19} = a^{32}\ell_0^{2^{n-1}} + a^{32} \ell_1 + (a^2 + a^{2^{n-1}})\ell_3^{4} + a^2 \ell_5  + a^{2^{n-1}}\ell_6^{2^{n-1}} + a^{32}\ell_{n-3}^{4} + a^{32}\ell_{n-1}^{4}.
\end{equation}

We now have all relations we need in order to prove our result.

\begin{proof}[Proof of Theorem~\ref{thm:trace}]
As we have already outlined above, we assume that $n$ is an odd integer with $n \geq 11$. Let us assume the existence of $a \in \F_{2^n} \setminus \{0\}$ and of a linearized polynomial $L(x) = \sum_{j=0}^{n-1}\ell_jX^{2^j} \in \F_{2^n}[X]$, such that $f$ (as defined in Equation~(\ref{eq:f})) equals the indicator of $\{0\}$, equal to $x^{2^n-1}+1\in \mathbb{F}_{2^n}[x]/(x^{2^n}+x)$.  Then, Relations~(\ref{eq:second}),~(\ref{eq:first}),~(\ref{eq:fourth}), and~(\ref{eq:fifth}) hold where the $v_r$ on the left-hand side are all equal to 0.

\paragraph{High level idea.} Our idea for simplifying the rather complex situation is to add to $L$ a polynomial satisfying Relation (\ref{eq:second}) for every $m$ and matching $L$ at as many coefficients as possible. The fact that this new polynomial has as few nonzero coefficients as possible will simplify the situation. Relation (\ref{eq:second}), if true for every $m$, writes: $\forall m, \ell_m+\ell_{m-1}^2=a^{2^{m-1}} (\ell_{m-1}+\ell_{m-2}^2)=\dots =a^{2^{m-1}+2^{m-2}+\dots +2}(\ell_1+\ell_0^2)= a^{2^m}\; \frac{\ell_1+\ell_0^2}{a^2}$.  A linearized polynomial $L(x)=\sum_{j=0}^{n-1}\ell_jx^{2^j}$ satisfies then Relation (\ref{eq:second}) for every $m\in \mathbb{Z}/n\mathbb{Z}$ if and only if  $\ell_m+\ell_{m-1}^2$ equals,  for every $m$, $a^{2^m}$ times a constant, that is, $L(x)+(L(x))^2=(b+b^2)\, tr(ax)$ for some $b$ (indeed, for every $x$, $L(x)+(L(x))^2$ has zero trace and this must then be the case also of the multiplicative constant). The relation $L(x)+(L(x))^2=(b+b^2)\, tr(ax)$ is equivalent to $L(x)=b\, tr(ax)+l(x)$ where $l(x)+(l(x))^2=0$, that is, $l$ is Boolean (and linear), that is, $l(x)=tr(cx)$ for some $c$. Having the choice of $b$ and $c$, we can try to choose them so that  $L(x)+b\, tr(ax)+tr(cx)$ has no term in $x^{2^3}$ and no term in $x^{2^4}$.

\paragraph{The proof in detail.} 
We observe that, for any $b,d \in \F_{2^n}$ and any Boolean function $\phi \colon \F_{2^n} \rightarrow \F_2$, adding $b \cdot tr(ax) +  \phi(x) + dx^2+d^2x$  to $L(x)$ does not change the function \[f(x)=(tr(x^{2^n-4})+1)(tr(ax)+1)(tr(x^{2^n-4}L(x))+1).\] Indeed, adding $b \cdot tr(ax)$ adds a function multiple of $tr(ax)$ to $(tr(x^{2^n-4}L(x))+1)$ which cancels with $(tr(ax)+1)$, adding $\phi(x)$ adds a function multiple of $tr(x^{2^n-4})$ to $(tr(x^{2^n-4}L(x))+1)$ which cancels with $(tr(x^{2^n-4})+1)$, and adding $dx^2+d^2x$ adds $dx^{2^n-2}+d^2x^{2^n-3}=dx^{2^n-2} + (dx^{2^n-2})^2$ to $(tr(x^{2^n-4}L(x))+1)$ which has zero trace and cancels. 

Therefore, if $L(x) = \sum_{i=0}^{n-1} \ell_i x^{2^i}$ is a linearized polynomial such that the coefficients $\ell_0,\dots,\ell_{n-1}$ satisfy all the Relations~(\ref{eq:second}),~(\ref{eq:first}),~(\ref{eq:fourth}), and~(\ref{eq:fifth}), then also, for every $b,c,d \in \F_{2^n}$, the coefficients $l_0, l_1, \dots, l_{n-1}$ of the polynomial 
\begin{align*} L'(x) &\coloneqq \sum_{i=0}^{n-1} l_i x^{2^i} =  b \cdot tr(ax) + tr(cx) + dx^2 + d^2x + \sum_{i=0}^{n-1} \ell_i x^{2^i} \\
&= (\ell_0 + ba + c + d^2)x + (\ell_1 + ba^2 + c^2 + d)x^2 + \sum_{i=2}^{n-1} (\ell_i + ba^{2^i} + c^{2^i})x^{2^i}
\end{align*}
satisfy all those relations and we can choose any $b,c,d \in \F_{2^n}$ without loss of generality. We split the remainder of this proof into two cases.

\paragraph{Case $tr(\frac{\ell_3}{a^8} + \frac{\ell_4}{a^{16}})= 0$ or $tr(\frac{\ell_4}{a^{16}} + \frac{\ell_5}{a^{32}})= 0$.}
 If $tr(\frac{\ell_3}{a^8} + \frac{\ell_4}{a^{16}})= 0$,  we can choose $b \in \F_{2^n}$ and $c \in \F_{2^n}$ such that   $\frac{\ell_3}{a^8}+b+\left(\frac ca\right)^8=0$ and $\frac{\ell_4}{a^{16}}+b+\left(\frac ca\right)^{16}=0$, i.e., $l_3 = \ell_3+b\, a^8+c^8= l_4 = \ell_4+b\, a^{16}+c^{16}=0$. We then have, by applying Relation (\ref{eq:second}) for $m \notin \{0,1,2,3,4,n-2,n-1\}$, that $l_{n-3}=l_{n-4}=\dots =l_4=l_3=0$.
 
 Similarly, if $tr(\frac{\ell_4}{a^{16}} + \frac{\ell_5}{a^{32}})= 0$, , we can choose $b \in \F_{2^n}$ and $c \in \F_{2^n}$ such that   $\frac{\ell_4}{a^{16}}+b+\left(\frac ca\right)^{16}=0$ and $\frac{\ell_5}{a^{32}}+b+\left(\frac ca\right)^{32}=0$, i.e., $l_4 = \ell_4+b\, a^{16}+c^{16}= l_5 = \ell_5+b\, a^{32}+c^{32}=0$. We then again have, by applying Relation (\ref{eq:second}) to the coefficients $l_m$ for $m \notin \{0,1,2,3,4,n-2,n-1\}$, that $l_{n-3}=l_{n-4}=\dots =l_4=l_3=0$.
 
 Further, by choosing $d \in \F_{2^n}$, we can assume without loss of generality\@ that $l_0 = 0$. Thus, we can restrict ourselves to functions $L'(x)=l_1x^2+l_2x^4+l_{n-2}x^{2^{n-2}}+l_{n-1}x^{2^{n-1}}$.

Relation (\ref{eq:first}), which is valid for $m\in \{6,\dots ,n-3\}$ gives, for $m=6, \dots, n-5$, 
\[
0 =  a^{2^{m-1}}l_2^{2^{m-1}} + l_2^{2^m} + a^{2^{m+1}}l_1^{2^{m}}  + a^{2^m} = (al_2 + l_2^2 + a^4l_1^2 + a^2)^{2^{m-1}}\]
and taking the $2^{m-1}$-th root yields $l_1^2 = \frac{l_2^2 + al_2 + a^2}{a^4}$. Applying  Relation (\ref{eq:first}) for $m=n-4$ gives \[
l_{n-2} = a^{2^{n-5}}l_2^{2^{n-5}} +  l_2^{2^{n-4}} + a^{2^{n-3}}l_1^{2^{n-4}}  + a^{2^{n-4}}\] and raising it to the $2^5$-th power yields $l_{n-2} = 0$. Similarly, applying Relation (\ref{eq:first}) for $m=n-3$ results in
\[ l_{n-1} = a^{2^{n-4}}l_2^{2^{n-4}} +  l_2^{2^{n-3}} + a^{2^{n-2}}l_1^{2^{n-3}}  + a^{2^{n-3}},\] thus also $l_{n-1} = 0$.

Now, applying Relation~(\ref{eq:fourth}) for $m \in \{4,5,\dots,n-3\}$ and substituting $l_1$ gives
\begin{equation}
\label{eq:final}
a^{2^m} l_2 + a^{2^m-4}l_2^2  + a^{2^m-3}l_2 + a l_2^{2^m}  = a^{2^m-2}.
\end{equation}
For $m = 4,5$ in Equation~(\ref{eq:final}), we obtain the system 
\begin{align*}
\begin{cases}a^{16} l_2 + a^{12}l_2^2  + a^{13}l_2 + a l_2^{16} &= a^{14}, \\
a^{32} l_2 + a^{28}l_2^2  + a^{29}l_2 + a l_2^{32} &= a^{30}.
\end{cases}
\end{align*}

Multiplying the first equation by $a^{16}$ and adding the second equation yields $a^{17} l_2^{16} + a l_2^{32} = 0$, i.e., $l_2 \in \{0,a\}$. This is a contradiction as neither of those values for $l_2$ satisfy the system.

 \paragraph{Case $tr(\frac{\ell_3}{a^8} + \frac{\ell_4}{a^{16}}) = tr(\frac{\ell_4}{a^{16}} + \frac{\ell_5}{a^{32}})= 1$.} In this case, we have $tr(\frac{\ell_3}{a^8} + \frac{\ell_5}{a^{32}}) = 0$. Because $n$ is odd, the image of the mapping $x \mapsto x + x^4$ over $\F_{2^n}$ is exactly the set $\{x \in \F_{2^n} \mid tr(x)=0\}$. Therefore,  we can choose $b \in \F_{2^n}$ and $c \in \F_{2^n}$ such that   $\frac{\ell_3}{a^8}+b+\left(\frac ca\right)^8=0$ and $\frac{\ell_5}{a^{32}}+b+\left(\frac ca\right)^{32}=0$, i.e., $l_3 = \ell_3+b\, a^8+c^8= l_5 = \ell_5+b\, a^{32}+c^{32}=0$. By applying Relation (\ref{eq:second}) to $m=5$, we obtain $l^2_4 = a^{2^4}l_4$, i.e., $l_4 = 0$ or $l_4 = a^{2^4}$. If $l_4=0$, then we can make the same calculations as in the previous case above, so we assume $l_4 = a^{2^4}$. We then obtain (also by applying Relation (\ref{eq:second})) that $l_6 = a^{2^6}$ and $l_7 = 0$, and inductively that $l_m = a^{2^m}$ for even values of $m \in \{4,5,\dots,n-3\}$ and $l_m = 0$ for odd values of $m \in \{4,5,\dots,n-3\}$. Further, by choosing $d \in \F_{2^n}$, we can assume without loss of generality\@ that $l_0 = 0$. 

Applying Relation~(\ref{eq:first}), which holds for $m \in \{6,7,\dots,n-3\}$, to $m = 6$ yields 
\begin{align*} 0 &= l_{m+2} + a^{2^{m-1}}l_m^{2^{n-1}} + l_m + a^{2^{m-1}}l_2^{2^{m-1}} + l_2^{2^m} + a^{2^{m+1}}l_1^{2^{m}}  + a^{2^m} \\
&= a^{2^{m-1}}l_2^{2^{m-1}} + l_2^{2^m} + a^{2^{m+1}}l_1^{2^{m}}  + a^{2^m} +a^{2^{m+2}}.\end{align*}
and by taking the $2^{m-1}$-th root, we obtain $0 = al_2 + l_2^2 + a^4l_1^2 + a^2 + a^8$. Applying Relation~(\ref{eq:first}) to $m=n-4$ (which is odd) gives
\begin{align*}
l_{n-2} &= a^{2^{m-1}}l_2^{2^{m-1}} + l_2^{2^m} + a^{2^{m+1}}l_1^{2^{m}} + a^{2^m} + a^{2^{m+2}} \\
&= (al_2 + l_2^2 + a^4l_1^2 + a^2 + a^8)^{2^{m-1}} = 0
\end{align*}
and applying Relation~(\ref{eq:first}) to $m=n-3$ (which is even) yields
\begin{align*}
l_{n-1} &= a^{2^{m-1}}l_2^{2^{m-1}} + l_2^{2^m} + a^{2^{m+1}}l_1^{2^{m}} + a^{2^{m+1}}l_0^{2^{m-1}} + a^{2^m} \\
&= (al_2 + l_2^2 + a^4l_1^2 + a^2 + a^8)^{2^{m-1}} + a^{2^{m+2}}= a^{2^{n-1}}.
\end{align*}

Now, applying Relation~(\ref{eq:fourth}) for $m \in \{4,5,\dots,n-3\}$ and substituting $l_1$ and $l_{n-1}$ yields
\begin{equation}
\label{eq:final_tr1}
(a^{2^m}+a^{2^{m}-3}) l_2 + a^{2^m-4}l_2^2  + a l_2^{2^m}  = a^{2^m-2}  + a^{2^{m+2}+1}.
\end{equation}
For $m = 4,5$ in Equation~(\ref{eq:final_tr1}), we obtain the system
\begin{align*}\begin{cases}(a^{16}+a^{13}) l_2 + a^{12}l_2^2  + a l_2^{16} &= a^{14} + a^{65}, \\
(a^{32}+a^{29}) l_2 + a^{28}l_2^2   + a l_2^{32} &= a^{30} + a^{129}.
\end{cases}
\end{align*}

Multiplying the first equation by $a^{16}$ and adding the second equation results in $a^{17} l_2^{16} + a l_2^{32} = a^{81} + a^{129}$, i.e., $l_2^2 + al_2 + (a^5 + a^8) = 0$ and the two solutions of this quadratic equation are $l_2 = a^4$ and $l_2 = a^4 + a$. Plugging those two values of $l_2$ into the first equation of the system yields $a^{17}=a^{14}$, that is, $a^3=1$, that is, $a\in \mathbb{F}_{2^{\gcd(2,n)}}^*$, that is, $a=1$ since $n$ is necessarily odd, in which case the system has $l_2 = 0$ and $l_2 = 1$ as the only solutions. In both cases, we have $l_1 = 0$. Relation~(\ref{eq:fifth}) writes $0=1+1+1$, since $v_{19}=l_0=l_1=l_5=0$ and $l_6=l_{n-1}=l_{n-3}=1$. This completes the proof by contradiction.
\end{proof}

\section{Concluding Remarks}
In the present paper, we have shown that the Gold APN and switched cube APN functions (including the APN family $x \mapsto x^3 + tr(x^9)$) in dimensions $n>5$  are not extendable to quadratic APN functions in one more variable having lowest possible (necessarily nonzero) nonlinearity. The question of the existence of quadratic APN functions in dimension $n > 7$ (necessarily in odd dimension) that admit an extension of such type remains open. We  know (see e.g. \cite{carlet2021boolean}) that the linearity of an $n$-variable quadratic Boolean function equals $2^\frac{n+k}2$, where $k$ equals the dimension of its linear kernel (that is, the vector space of the directions of its derivatives that are constant). We know that quadratic APN functions in odd dimension are almost bent (AB), that is, have linearity $2^\frac{n+1}2$. A quadratic $n$-variable AB function is then extendable to a quadratic $(n+1)$-variable function $F$ with linearity $2^n$ if some component function of $F$ has a linear kernel of dimension $n-1$, while every component function of the restriction of $F$ to a hyperplane (valued in a hyperplane) has a linear kernel of dimension 1. We leave the question open whether this is possible for large enough values of $n$. Even the question  of the existence of known quadratic APN functions in dimension $n>7$ that admit an extension of such type remains open.  A first task needed to be achieved for progressing on this problem would be to determine the ortho-derivatives of the other known classes of APN functions (see a list of known AB functions up to equivalence in e.g. \cite[Chapter 11]{carlet2021boolean}). There are in particular binomials, whose gamma functions are obtained in~\cite{DBLP:conf/itw/BudaghyanCH11}, but only up to a conjecture, and depend on a function that is not explicit. There are hexanomials, whose case seems more accessible. The more general problem of determining whether APN functions with low nonlinearity exist for any $n$ large enough is a wide open question, with huge interest for our knowledge on general APN functions.  

\subsubsection*{Acknowledgments}
We thank the reviewers for their detailed and useful comments, which helped improving the paper.

The research of the first author is supported by Deutsche  Forschungsgemeinschaft  (DFG) under Germany's Excellence Strategy - EXC 2092 CASA - 390781972. 

The research of the second author is partly supported by the Trond Mohn Foundation and Norwegian Research Council.

\end{document}